\documentclass[11pt]{article}

\usepackage{fullpage}

\usepackage{amsmath,amsfonts,amssymb,amsthm,url}
\usepackage{graphics}
\usepackage{epsfig,color}

\newtheorem{theorem}{Theorem}[section]
\newtheorem{lemma}[theorem]{Lemma}
\newtheorem{definition}[theorem]{Definition}
\newtheorem{remark}[theorem]{Remark}

\def\N{{\mathbb N}}
\def\reals{{\mathbb R}}
\def\eps{{\varepsilon}}
\def\bd{{\partial}}

\def\A{{\cal A}}

\def\td{{\sf td}}
\def\mx{{\rm max}}
\def\vol{{\rm Vol}}

\def\conv{{\rm conv}}

\def\Pr{{\bf Pr}}

\begin{document}

\bibliographystyle{plainurl}
\title{How to Find a Point in the Convex Hull Privately}

\author{Haim Kaplan\thanks{School of Computer Science, Tel Aviv University, Tel~Aviv, and Google,
haimk@tau.ac.il.
Partially supported by ISF grant 1595/19 and grant 1367/2016
   from the German-Israeli Science Foundation (GIF)}
\and
Micha Sharir\thanks{School of Computer Science, Tel Aviv University, Tel~Aviv,
   Israel, michas@tau.ac.il. Partially supported by ISF Grant 260/18,  by grant 1367/2016
   from the German-Israeli Science Foundation (GIF), and by
   Blavatnik Research Fund in Computer Science at Tel Aviv University.}
\and
Uri Stemmer\thanks{Department of Computer Science, Ben-Gurion University, and Google, u@uri.co.il.
Partially supported by ISF grant 1871/19.}%
}


\maketitle

\begin{abstract}
We study the question of how to compute a point in the convex hull of an input set $S$ of $n$ points in $\reals^d$ in a differentially private manner. This question, which is trivial without privacy requirements, turns out to be quite deep when imposing differential privacy. In particular, it is known that the input points must reside on a fixed {\em finite} subset $G\subseteq\reals^d$, and furthermore, the size of $S$ must grow with the size of $G$. Previous works~\cite{BKN10,BeimelMNS19,BNS13b,BunDRS18,BNSV15,KLMNS19} focused on understanding how $n$ needs to grow
 with $|G|$, and showed that $n=O\left( d^{2.5} \cdot 8^{\log^*|G|} \right)$ suffices (so $n$ does not have to grow significantly
  with $|G|$). However, the available constructions exhibit running time at least $|G|^{d^2}$, where typically $|G|=X^d$ for some (large) discretization parameter $X$, so the running time is in fact $\Omega(X^{d^3})$.

In this paper we give a differentially private algorithm that runs in $O(n^d)$ time, assuming that
$n=\Omega(d^4 \log X)$. To get this result we study and exploit some structural properties of the Tukey levels (the regions $D_{\ge k}$ consisting of points whose Tukey depth is at least $k$, for $k=0,1,\dots$). In particular, we derive lower bounds on their volumes for point sets $S$ in general position, and develop a rather subtle mechanism for handling point sets $S$ in degenerate position (where
the deep Tukey regions have zero volume). A naive approach to the construction of the Tukey regions requires $n^{O(d^2)}$ time. To reduce
the cost to $O(n^d)$, we use an approximation scheme for estimating the volumes of the Tukey regions (within their affine spans in case of degeneracy), and for sampling a point from such a region, a scheme that is based on the volume estimation framework of Lov\'asz and Vempala~\cite{LV:vol} and of Cousins and Vempala~\cite{CV}. Making this framework differentially private raises a set of technical challenges that we  address.
\end{abstract}


\section{Introduction}
We often would like to analyze data  while protecting the privacy of the individuals that contributed to it.
At first glance, one might hope to ensure privacy by simply deleting all names and ID numbers from
the data. However, such anonymization schemes 
are proven time and again to violate privacy.
This gave rise of a theoretically-rigorous line of work that has placed private data analysis on firm foundations, centered around a mathematical definition for privacy known as {\em differential privacy}~\cite{DMNS06}.

Consider a database $S$ containing personal information of individuals. Informally, an algorithm operating on such a database is said to preserve differential privacy if its outcome distribution is (almost) insensitive to any arbitrary change to the data of one individual in the database. Intuitively, this means that an observer looking at the outcome of the algorithm (almost) cannot distinguish between whether Alice's information is $x$ or $y$ (or whether Alice's information is present in the database at all) because in any case it would have (almost) no effect on the outcome distribution of the algorithm.

\begin{definition}[Dwork et al.~\cite{DMNS06}]\label{def:DP}
Two databases (multisets) $S$ and $S'$ are called {\em neighboring} if  they differ in a single entry.
That is, $S = S_0\cup \{x\}$ and $S'=S_0\cup \{y\}$ for some items $x$ and $y$. A randomized algorithm $A$ is $(\varepsilon,\delta)$-{\em differentially private} if for every two neighboring databases $S,S'$ and for any event $T$ we have
$$
\Pr[A(S)\in T]\leq e^{\varepsilon}\cdot\Pr[A(S')\in T]+\delta.
$$
When $\delta=0$ this notion is referred to as {\em pure} differential privacy, and when $\delta>0$ it is referred to as {\em approximate} differential privacy.
\end{definition}

\begin{remark}
Typically, $\eps$ is set to be a small constant, say $\eps=0.1$, and $\delta$ is set to be a small function of the database size $|S|$ (much smaller than $1/|S|$). Note that to satisfy the definition (in any meaningful way) algorithm $A$ must be randomized.
\end{remark}

Differential privacy is increasingly accepted as a standard for rigorous treatment of
privacy. However, even though the field has witnessed an explosion of research in the recent years, much remains unknown and answers to fundamental questions are still missing. In this work we study one such fundamental question, already studied in~\cite{BKN10,BeimelMNS19,BNS13b,BunDRS18,BNSV15,KLMNS19}: Given a database containing points in $\reals^d$ (where every point is assumed to be the information of one individual), how can we {\em privately} identify a point in the {\em convex hull} of the input points? 
This question, which is trivial without privacy requirements, turns out to be quite deep when imposing differential privacy. In particular, Bun et al.~\cite{BNSV15} showed that in order to be able to solve it, we must assume that the input points reside on a fixed {\em finite} subset $G\subseteq\reals^d$, and furthermore, the number of input points must grow with the size of $G$.

\medskip
\setlength{\fboxsep}{10pt}
\noindent\fbox{\parbox{0.95\textwidth}{
{\bf The Private Interior Point (PIP) Problem.}\\ 
Let $\beta,\varepsilon,\delta,X$ be positive
 parameters where $\beta,\varepsilon,\delta$ are small and $X$ is a large integer. 
Let $G\subseteq[0,1]^d$ be a finite uniform grid with side steps $1/X$ (so $|G|=(X+1)^d$).  
Design an algorithm $A$ such that for some $n\in\N$ (as small as possible as a function of $\beta,\varepsilon,\delta,X$) we have
\begin{enumerate}
	\item {\bf Utility:} For every database $S$ containing at least $n$ points from $G$ it holds that $A(S)$ returns a point in the convex hull of $S$ with probability at least $1-\beta$. (The outcome of $A$ does not have to be in $G$.)
	\item {\bf Privacy:} For every pair of neighboring databases $S,S'$, each containing $n$ points from $G$, and for any event $T$, we have
	$\Pr[A(S)\in T]\leq e^\eps \cdot\Pr[A(S')\in T]+\delta.$
\end{enumerate}
}}
\medskip

The parameter $n$ is referred to as {\em the sample complexity} of the algorithm. 
It is the smallest number of points on which we are guaranteed to succeed (not to be confused with the actual size of the input).
The PIP problem is very natural on its own. Furthermore, as was observed in~\cite{BeimelMNS19}, an algorithm for solving the PIP problem can be used as a building block in other applications with differential privacy, such as learning halfspaces and linear regression. 
Previous works~\cite{BKN10,BeimelMNS19,BNS13b,BunDRS18,BNSV15,KLMNS19} have focused on the task of minimizing the sample complexity $n$ while ignoring the runtime of the algorithm. In this work we seek an efficient algorithm for the PIP problem, that still keeps the sample complexity $n$ ``reasonably small'' (where ``reasonably small'' will be made precise after we introduce some additional notation). 

\subsection{Previous Work}

Several papers studied the PIP problem for  $d=1$.  
In particular, three different constructions with sample complexity $2^{O(\log^*|G|)}$ were presented in~\cite{BNS13b,BunDRS18,BNSV15} (for $d=1$). Recently, Kaplan et al.~\cite{KLMNS19} presented a new construction with sample complexity $O((\log^*|G|)^{1.5})$ (again, for $d=1$).
Bun et al.~\cite{BNSV15} gave a lower bound showing that every differentially private algorithm for this task must have sample complexity $\Omega(\log^*|G|)$. Beimel et al.~\cite{BeimelMNS19} incorporated a dependency in $d$ to this lower bound, and showed that every differentially private algorithm for the PIP problem must use at least $n=\Omega(d+\log^*|G|)$ input points.

For the case of {\em pure} differential privacy (i.e., $\delta=0$), a lower bound of $n=\Omega(\log X)$ on the sample complexity follows from the results of Beimel et al.~\cite{BKN10}. This lower bound is tight, as an algorithm with sample complexity $n=O(\log X)$ (for $d=1$) can be obtained using a generic tool in the literature of differential privacy, called the {\em exponential mechanism}~\cite{MT07}. We sketch this application of the exponential mechanism here. (For a precise presentation of the exponential mechanism see Section \ref{sec:exp}.) Let $G=\{0,\frac{1}{X},\frac{2}{X},\dots,1\}$ be our (1-dimensional) grid within the interval $[0,1]$, and let $S$ be a multiset containing $n$ points from $G$. The algorithm is as follows.
\begin{enumerate}
\item For every $y\in G$ define the {\em score} $q_S(y)=\min\left\{ |\{x\in S \mid x\geq y\}|,\; |\{x\in S \mid x\leq y\}|  \right\}.$
\item Output $y\in G$ with probability proportional to $e^{\eps\cdot q_S(y)}$.
\end{enumerate}

Intuitively, this algorithm satisfies differential privacy because changing one element of $S$ changes the score $q_S(y)$ by at most $\pm1$, and thus changes the probabilities with which we sample elements by roughly an $e^{\eps}$ factor. As for the utility analysis, observe that $\exists y\in G$ with $q_S(y)\geq\frac{n}{2}$, and the probability of picking this point is (at least) proportional to $e^{\eps n/2}$. As this probability increases exponentially with $n$, by setting $n$ to be big enough we can ensure that points $y'$ outside of the convex hull (those with $q_S(y')=0$) get picked with very low probability.

Beimel et al.~\cite{BeimelMNS19} observed that this algorithm extends to higher dimensions by replacing $q_S(y)$ with the {\em Tukey depth} $\td_S(y)$ of the point $y$ with respect to the input set $S$ (the Tukey depth of a point $y$ is the minimal number of points that need to be removed from $S$ to ensure that $y$ is {\em not} in the convex hull of the remaining input points). However, even though there must exist a point $y\in\reals^d$ with high Tukey depth (at least $n/(d+1)$; see~\cite{Matousek:2002}), the finite grid $G\subseteq\reals^d$ might fail to contain such a point. Hence, 
Beimel et al.~\cite{BeimelMNS19} first {\em refined} the grid $G$ into a grid $G'$ that contains a point with high Tukey depth, and then randomly picked a point $y$ from $G'$ with probability proportional to $e^{\eps\cdot\td_S(y)}$. To compute the probabilities with which grid points are sampled, the algorithm in~\cite{BeimelMNS19} explicitly computes the Tukey depth of every point in $G'$, which, because of the way in which $G'$ is defined, results in running time of at least $\Omega(|G|^{d^2})=\Omega(X^{d^3})$ and sample complexity $n=O(d^3 \log |G|)=O(d^4 \log X)$. Beimel et al.\ then presented an improvement of this algorithm with reduced sample complexity of $n=O(d^{2.5}\cdot8^{\log^*|G|})$, but the running time remained $\Omega(|G|^{d^2})=\Omega(X^{d^3})$.   

\subsection{Our Construction}

We give an approximate differentially private algorithm
for the private interior point problem that runs in $O(n^d)$  time,\footnote{When we use  $O$-notation for time complexity we  hide logarithmic factors in $X$, $1/\eps$, $1/\delta$, and polynomial factors in $d$. We assume operations on real numbers in $O(1)$ time (the so-called real RAM model).}  and succeeds with high probability when the size of its input is $\Omega(\frac{d^4}{\eps} \log \frac{X}{\delta})$.
Our algorithm is obtained by carefully implementing the exponential mechanism and reducing its running time from $\Omega(|G|^{d^2})=\Omega(X^{d^3})$ to $O(n^d)$. We now give an informal overview of this result. 

To avoid the need to extend the grid and to calculate the Tukey depth of each point in the extended grid, we sample our output directly from $[0,1]^d$. To compute the probabilities with which we sample a point from $[0,1]^d$ we compute, for each $k$ in an appropriate range, the {\em volume} of the Tukey region of depth $k$, which we denote as $D_k$. (That is, $D_k$ is the region in $[0,1]^d$ containing all points with Tukey depth exactly $k$.) We then sample a value $k\in[n]$ with probability proportional to $\vol(D_k)\cdot e^{\eps k}$, and then sample a random point uniformly from $D_k$.

Observe that this strategy picks a point with Tukey depth $k$ with probability proportional to $\vol(D_k)\cdot e^{\eps k}$. Hence, if for a ``large enough'' value of $k$ (say $k\geq \frac{n}{cd}$ for a suitable absolute constant $c>1$) we have that $\vol(D_k)$ is ``large enough'', then a point with Tukey depth $k$ is picked with high probability. 
However, if $\vol(D_k)=0$ (or too small) then a point with Tukey depth $k$ is picked with probability zero (or with too small a probability). Therefore, to apply this strategy, we derive a lower bound on the volume of every non-degenerate Tukey region, showing that if the volume is non-zero, then it is at least $\Omega\left(1/X^{d^3}\right)$.

There are two issues here. The first issue is that
the best bound we know on the complexity of a
 Tukey region is
$O( n^{(d-1)\lfloor d/2 \rfloor})$, so  we cannot compute these regions explicitly (in the worst-case) in time $O(n^d)$ (which is our target runtime complexity). We avoid the need to compute the Tukey regions explicitly by using an approximation scheme for estimating the volume of each region and for sampling a point from such a region, a scheme that is based on the volume estimation framework of Lov\'asz and Vempala~\cite{LV:vol} and of Cousins and Vempala~\cite{CV}. The second issue is that it might be the case that all Tukey regions for large values of $k$ are degenerate, i.e., have volume 0, in which case this strategy might fail to identify a point in the convex hull of the input points.

\medskip
\noindent{\bf Handling degeneracies.} 
We show that if the Tukey regions of high depth are degenerate, then many of the input points must lie in a lower-dimensional affine subspace.  
This can be used to handle degenerate inputs $S$ as follows.
We first (privately) check whether there exists an affine proper subspace that
contains many points of $S$. If we find such a subspace $f$, we recursively
continue the procedure within $f$, with respect to $S\cap f$.
Otherwise, if no such subspace exists, then it must be the case that the Tukey regions of high depth are full-dimensional (with respect to the subspace into which we have recursed so far), so we can apply our algorithm for the non-degenerate case and obtain a point that lies, with high probability, in the convex hull of the surviving subset of $S$, and thus of the full set $S$.

We remark that it is easy to construct algorithms with running time polynomial in the input size $n$, when $n$ grows exponentially in $d$. (In that case one can solve the problem using a reduction to the 1-dimensional case.)
In this work we aim to reduce the running time while keeping the sample complexity $n$ at most polynomial in $d$ and in $\log|G|$.


    \section{Preliminaries}

   We assume that our input set $S$ consists of $n$ points that lie in the intersection of a grid $G$
    with the cube $[0,1]^d$ in $\reals^d$.
    We assume that $G$ is of
    side length $1/X$ for a given accuracy integer parameter $X$, so it
    partitions the cube into $X^d$ cells.

\subsection{The exponential mechanism}
\label{sec:exp}
Let $G^*$ denote the set of all finite databases over a grid $G$, and let $F$ be a finite set. Given a database $S\in G^*$, a {\em quality} (or {\em scoring}) function $q:G^*\times F \rightarrow \N$ assigns a number $q(S,f)$ to each element  $(S,f)\in G^*\times F$, identified as the ``quality'' of $f$ with respect to $S$. We say that the function $q$ has {\em sensitivity} $\Delta$ if for all neighboring databases $S$ and $S'$ and for all $f\in F$ we have $|q(S,f)-q(S',f)|\leq \Delta$.

The \emph{exponential mechanism} of McSherry and Talwar~\cite{MT07} privately identifies an element $f\in F$ with large quality $q(S,f)$. Specifically, it chooses an element $f\in F$ randomly, with probability proportional to $\exp\left(\eps \cdot q(S,f) /(2 \Delta)\right)$. The privacy and utility of the mechanism are:

\begin{theorem}[McSherry and Talwar \cite{MT07}] \label{prop:exp_mech}
The exponential mechanism is $(\eps,0)$-differentially private.
Let $q$ be a quality function with sensitivity $\Delta$. Fix a database $S \in G^*$ and let ${\rm OPT} = \max_{f\in F}\{q(S,f)\}$.
For any $\beta\in (0,1)$, with
 probability at least $(1-\beta)$, the exponential mechanism outputs a solution $f$ with quality
$q(S,f)\geq{\rm OPT}-\frac{2\Delta}{\eps}\ln\frac{|F|}{\beta}$.
\end{theorem}

    \subsection{Tukey depth}

    The \emph{Tukey depth} \cite{Tukey} $\td_S(q)$ of a point $q$
    with respect to $S$ is the minimum number of points of $S$ we need to
    remove to make $q$ lie outside the convex hull of the remaining subset.
    Equivalently, $\td_S(q)$ is the smallest number of points
    of $S$ that lie in a closed halfspace containing $q$.
    We will write $\td(q)$ for $\td_S(q)$ when the set $S$ is clear from the context.
    It easily follows from Helly's theorem that there is always a point of
    Tukey depth at least $n/(d+1)$ (see, e.g., \cite{Matousek:2002}).
        We denote the largest Tukey depth of a point by $\td_\mx(S)$ (the maximum Tukey depth is always at most $n/2$).

    We define the regions
    $
    D_{\ge k}(S) = \left\{ q\in [0,1]^d \mid \td_S(q \ge k \right\}$ and $D_{k}(S) = D_{\ge k}(S) \setminus D_{\ge k+1}(S)$
    for $k = 0,\ldots,\td_\mx(S)$.
    Note that $D_{\ge 1}$ is the convex hull of $S$ and that $D_{\ge 0} = [0,1]^d$.
    It is easy to show that $D_{\ge k}$ is convex; $D_{\ge k}$ is in fact the intersection
    of all (closed) halfspaces containing at least $n-k+1$ points of $S$; see \cite{RR98}.
        It is easy to show that all this is true also when $S$ is degenerate.
See Figure~\ref{layers} for an illustration. The following lemma is easy to establish.

    \begin{lemma}
    \label{lem:non-deg-D}
    \label{non-deg-D}
    If $D_{\ge k}$ is of dimension $d$ (we refer to such a region as non-degenerate)
    then
    $C_k = \bd D_{\ge k}(S)$ is a convex polytope, each of whose facets
    is contained in  a simplex $\sigma$
    spanned by $d$ points of $S$, such that one of the open halfspaces bounded
    by the hyperplane supporting $\sigma$ contains exactly $k-1$ points of $S$.
    \end{lemma}

\begin{figure}[htb]
\begin{center}
\begin{picture}(0,0)%
\includegraphics{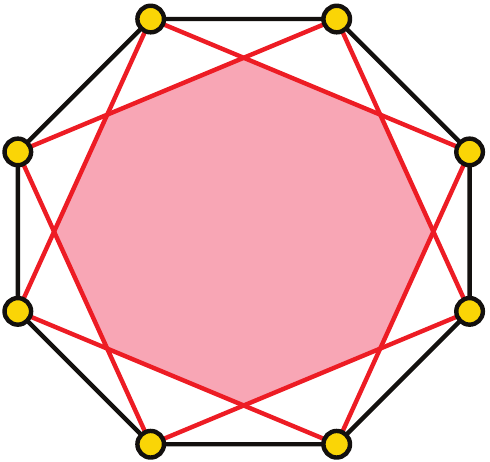}%
\end{picture}%
\setlength{\unitlength}{5594sp}%
\begingroup\makeatletter\ifx\SetFigFont\undefined%
\gdef\SetFigFont#1#2#3#4#5{%
  \reset@font\fontsize{#1}{#2pt}%
  \fontfamily{#3}\fontseries{#4}\fontshape{#5}%
  \selectfont}%
\fi\endgroup%
\begin{picture}(1650,1560)(2011,-1471)
\put(2611,-736){\makebox(0,0)[lb]{\smash{{\SetFigFont{14}{16.8}{\rmdefault}{\mddefault}{\updefault}{\color[rgb]{0,0,0}$D_{\ge 2}$}%
}}}}
\end{picture}%
\caption{The Tukey layers $D_{\ge 2}$ and $D_{\ge 1}$.}
\label{layers}
\end{center}
\end{figure}
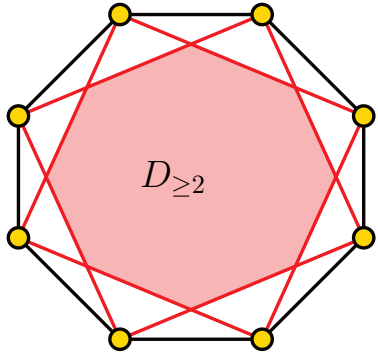

    \section{The case of general position}
    \label{sec:base}
     \label{sec:gen}
    \label{sec:no-deg}
		\label{uni-gen}
		
As already said, we apply the exponential mechanism for privately identifying a point in $[0,1]^d$ with (hopefully) large Tukey depth with respect to the input set $S$. This satisfies (pure) differential privacy since the sensitivity of the Tukey depth is $1$. In this section we show that when the input points are in general position, then this application of the exponential mechanism succeeds (with high probability) in identifying a point that has positive Tukey depth, that is, a point inside the convex hull of $S$.

To implement the exponential mechanism (i.e., to sample a point from $[0,1]^d$ appropriately), we need to compute the Tukey regions $D_{\ge k}$ and their volumes. In this section we compute these regions explicitly in $O\left( n^{1+(d-1) \lfloor d/2  \rfloor} \right)$ time. In Section~\ref{sec:nd_time} we will show that the cost can be reduced to $O(n^d)$.

\paragraph*{Computing ${\boldsymbol{D_{\ge k}}}$.}
We pass to the dual space, and construct the
    arrangement $\A(S^*)$ of the hyperplanes dual to the points of $S$.
    A point $h^*$ dual to a hyperplane $h$ supporting $D_{\ge k}$ has at least $n-k+1$
dual hypeplanes passing below $h^*$ or incident to $h^*$, or, alternatively, passing above
$h^*$ or incident to $h^*$. Furthermore, if we move $h^*$ slightly down in the first case
(resp., up in the second case), the number of hypeplanes below (resp., above) it becomes
smaller than $n-k+1$. 

When $h^*$ is a vertex 
 of $\A(S^*)$ we refer to it  as \emph{$k$-critical},
or simply as \emph{critical}. If $D_{\ge k}$ is non-degenerate then, by Lemma \ref{non-deg-D},
    each hyperplane $h$ that supports a facet of $D_{\ge k}$ must be spanned by $d$
affinely independent points of $S$. That is, all these hyperplanes are dual to $k$-critical vertices
of $\A(S^*)$.

    We compute all critical dual vertices
    that have at least $n-k+1$ dual hyperplanes passing below them or incident to them,
or those with at least $n-k+1$ dual hyperplanes passing above them or incident to them.
    The intersection of the appropriate primal halfspaces that are bounded by the hyperplanes
corresponding to these dual vertices 
is $D_{\ge k}$. This gives  an algorithm for constructing $D_{\ge k}$ when it is  non-degenerate.
    Otherwise $D_{\ge k}$ is degenerate and its volume is $0$, and we need additional
techniques, detailed in the next subsection, to handle such situations.

    We compute the volume of each non-degenerate $D_{\ge k}$, for
    $k=1,\ldots, \td_\mx(S)$. We do that in brute force, by computing and triangulating
    $D_{\ge k}$ and adding up the volumes of the simplices in this triangulation.
    Then we subtract the volume of $D_{\ge k+1}$ from the volume of $D_{\ge k}$
to get the volume of $D_{k}$.

\paragraph*{The sampling mechanism.}
    We assign to each $D_k$ the weight $e^{\eps k/2}$
and sample a region $D_k$, for $k=0,\ldots, \td_\mx(S)$, with probability
    \[
    \mu_k = \frac{e^{\eps k/2}{\rm Vol}(D_k)}
    {\sum_{j\ge 0} e^{\eps j/2} {\rm Vol}(D_j)} ,
    \]
    where ${\rm Vol}(D_k)$ denotes the volume of $D_k$.
Then we sample a point uniformly at random from
$D_k$. We do this in brute force by computing $D_k$, triangulating it, computing the volume of
each simplex, drawing a simplex from this triangulation with probability proportional to
its volume, and then drawing a point uniformly at random from the chosen simplex.\footnote{A simple way to implement the last step is to draw uniformly and independently
$d$ points from $[0,1]$, compute the lengths $\lambda_1,\ldots,\lambda_{d+1}$ of the
intervals into which they partition $[0,1]$, and return $\sum_{i=1}^{d+1} \lambda_i v_i$,
where $v_1,\ldots,v_{d+1}$ are the vertices of the simplex.}

This is an instance of the exponential mechanism in which the score (namely the Tukey depth)
has sensitivity $1$, i.e., $|\td_S(q) - \td_{S'}(q)| \le 1$ for any point $q\in [0,1]^d$,
when $S$ and $S'$ differ by only one element. It thus follows from 
the properties of the exponential mechanism (Theorem~\ref{prop:exp_mech})
that this procedure is (purely) $\eps$-differentially private.

\paragraph*{Complexity.}
Computing the dual arrangement $\A(S^*)$ takes $O(n^d)$ time~\cite{ESS}.
Assume that $D_{\ge k}$ is non-degenerate and let $M_k$ denote the number of
hyperplanes defining $D_{\ge k}$ (i.e., the hyperplanes supporting its facets).
It takes $O(M_k^{\lfloor d/2\rfloor})$ time to construct $D_{\ge k}$, as the
intersection of $M_k$ halfspaces, which is a dual version of constructing the
convex hull (see~\cite{Chaz}). Within the same asymptotic bound we can triangulate
$D_{\ge k}$ and compute its volume.
We obviously have $M_k = O(n^d)$, but the number can be somewhat reduced. The following lemma is known. 
\begin{lemma}[Proposition 3 in \cite{Liu19}] \label{lem:Dd-1}
    The number of halfspaces needed to construct $D_{\ge k}$ is $O(n^{d-1})$.
\end{lemma}
\begin{proof}
Fix a $(d-1)$-tuple $\sigma$ of points of $S$,
and consider all the relevant (closed) halfspaces, each of which is
bounded by a hyperplane that is spanned by $\sigma$ and another point of $S$,
and contains at least $n-k+1$ points of $S$. It is easy to check
that, as long as the intersection of these halfspaces is full-dimensional, it is equal
to the intersection of just two of them.
\end{proof}

Summing up, we get that computing the volume of all the non-degenerate regions $D_{\ge k}$,
for $k=1,\ldots, \td_\mx(S)$, takes
$O\left( \sum_{k\ge 1} M_k^{\lfloor d/2\rfloor} \right) =
O\left( n^{1+(d-1) \lfloor d/2  \rfloor} \right)$
time. 

\paragraph*{Utility.}
    We continue to assume that $D_{\ge k}$ is non-degenerate, and give a lower bound on its volume.
    By Lemma \ref{lem:non-deg-D}, such a $D_{\ge k}$ is the intersection of
    halfspaces, each bounded by a hyperplane that is spanned by
    $d$ points of $S$. Denote by $H$ the set of these hyperplanes.
    To obtain the lower bound, it suffices to consider the case where $D_{\ge k}$
is a simplex, each of whose vertices is the intersection point of $d$ hyperplanes of $H$.

    The equation of a hyperplane $h$ that passes through $d$ points,
    $a_1,\ldots,a_d$, of $S$ is
\[
    \begin{vmatrix}
    1 & x_1 & \cdots & x_d \\
    1 & a_{1,1} & \cdots & a_{1,d} \\
    & & \vdots & \\
    1 & a_{d,1} & \cdots & a_{d,d}
    \end{vmatrix} = 0 ,
\]
    where $a_i = (a_{i,1},\ldots,a_{i,d})$, for $i=1,\ldots,d$.
    The coefficients of the $x_i$'s in the equation of $h$ are
    $d\times d$ subdeterminants of this determinant, where each determinant has
    one column of $1$'s, and $d-1$ other columns,
    each of whose entries is some $a_{i,j}$, which is a rational of the
    form $m/X$, with $0\le m\le X$ (the same holds for the $1$'s, with $m=X$).
    The value of such a determinat
    (coefficient) is a rational number with denominator $X^{d}$.
    By Hadamard's inequality, its absolute value is at most the product
    of the Euclidean norms of its rows, which is at most $d^{d/2}$.
    That is, the numerator of the determinant is an integer of absolute value at most $d^{d/2}X^{d}$.
    The free term is a similar sub-determinant, but all its entries are the
    $a_{i,j}$'s, so it too is a rational with denominator $X^d$, and with
    numerator of absolute value at most $d^{d/2}X^d$. Multiplying by $X^{d}$, all the
    coefficients become integers of absolute value at most $d^{d/2}X^d$.

    Each vertex of any region $D_k$ of Tukey depth $k$, for any $k$,
    is a solution of a linear system of $d$ hyperplane equations
    of the above form. It is therefore a rational number whose denominator,
    by Cramer's rule, is the determinant of all non-free coefficients
    of the $d$ hyperplanes. This is an integer whose absolute value,
    again by Hadamard's inequality, is at most
\[
    \left( \sqrt{d} d^{d/2}X^{d} \right)^d \le d^{d(d+1)/2} X^{d^2} .
\]
    Since the free terms are also integers, we conclude that the
    coordinates of the intersection point are rationals with a common
    integral denominator of absolute value at most $d^{d(d+1)/2}X^{d^2}$.

    We can finally obtain a lower bound for the nonzero volume of a simplex spanned by
    any $d+1$ linearly independent intersection points $v_1,\ldots,v_{d+1}$.
    This volume is
\[
    \frac{1}{d!}
    \begin{vmatrix}
    1 & v_{1,1} & \cdots & v_{1,d} \\
    1 & v_{2,1} & \cdots & v_{2,d} \\
    & & \vdots & \\
    1 & v_{d+1,1} & \cdots & v_{d+1,d}
    \end{vmatrix},
    \]
        where again $v_i = (v_{i,1},\ldots,v_{i,d})$, for $i=1,\ldots,d+1$.
    Note that all the entries in any fixed row have the same denominator.
    The volume is therefore a rational number whose denominator is $d!$ times
    the product of these denominators, which is thus an integer with
    absolute value at most
\[
    d! \cdot \left( d^{d(d+1)/2}X^{d^2} \right)^d \le (dX)^{d^3}
\]
    (for $d\ge 2$).
    That is, we get
    the following lemma.
    \begin{lemma}
        If the volume of $D_{\ge k}$ is not zero then it is at least $1/(dX)^{d^3}$.
    \end{lemma}


    Assume that the volume of $D_{\ge k}$ is not zero for $k=k_0:=n/(4d)$.
    Since the score of a point outside the convex hull is zero and the volume of $D_{\ge 0}$
is at most $1$, we get that the probability to sample a point outside of the convex hull is at most
    \[
    \frac{1}{e^{\eps k_0}{\rm Vol}(D_{k_0})}
    \le \frac{{(dX)^{d^3}}}{e^{\eps n/(4d)}}.
    \]
    This inequality leads to
    the following theorem, which summarizes the utility that our instance of the exponential mechanism provides.
    \begin{theorem} \label{thm:exp-mechanism}
    If ${\displaystyle n \ge \frac{4d^4\log (dX)}{\eps} + \frac{4d}{\eps}\log\frac{1}{\beta}}$
    and $D_{\ge n/4d}$ has non-zero volume then the exponential mechanism, implemented as above,
returns a point in the convex hull with probability at least $1-\beta$, in $O\left( n^{1+(d-1) \lfloor d/2  \rfloor} \right)$ time.
    \end{theorem}


    \section{Handling degeneracies} \label{uni-deg}
    In general we have no guarantee that $D_{\ge n/4d}$ has non-zero volume.
    In this section we show how to overcome this and compute (with high probability) a point in the convex hull of any input.
 We rely on the
    following lemma, which shows that if
    $D_{\ge k}$ has zero volume then many points of $S$ are in a lower-dimensional affine subspace.
    \begin{lemma} \label{lem:no-vol} \label{flat}
    If $D_{\ge k}$ spans an affine subspace $f$ of dimension $j$ then
\[
|S\cap f| \ge n-(d-j+1)(k-1) .
\]
    \end{lemma}

    \begin{proof}
    Recall that $D_{\ge k}$ is the intersection of all closed halfspaces $h$ that contain
at least $n-k+1$ points of $S$.
Note that a halfspace that bounds $D_{\ge k}$ and whose bounding
hyperplane properly crosses $f$, defines a proper halfspace within $f$, and, by assumption, the
intersection of these halfspaces has positive relative volume. This means that the intersection
of these halfspaces in $\reals^d$ has positive volume too, and thus cannot confine $D_{\ge k}$ to $f$.
To get this confinement, there must exist (at least) $d-j+1$ halfspaces in the above collection,
whose intersection is $f$. Hence the union of their complements is $\reals^d\setminus f$.
    Since this union contains at most $(d-j+1)(k-1)$ points of $S$, the claim follows.
    \end{proof}

In what follows, to simplify the expressions that we manipulate, we use the weaker lower bound
$n-(d-j+1)k$.
    In order for the lemma to be meaningful, we want $k$ to be significantly smaller than
    the centerpoint bound $n/(d+1)$, so we set, as above, $k = n/(4d)$.

    We use Lemma \ref{lem:no-vol} to handle degenerate inputs $S$, using the following high-level approach.
    We first (privately) check whether there exists an affine proper subspace that
    contains many points of $S$. If we find such a subspace $f$, we recursively
    continue the procedure within $f$, with respect to $S\cap f$.
    Lemma~\ref{flat} then implies that we do not lose too many points when we recurse
    within $f$ (that is, $|S\cap f|$ is large), using our choice $k= n/(4d)$.
    Otherwise, if no such subspace exists, Lemma \ref{lem:no-vol} implies that $D_{\ge k}$
is full-dimensional (with respect to the subspace into which we have recursed so far),
so we can apply the exponential mechanism, as implemented in Section \ref{sec:no-deg}, and
obtain a point that lies, with high probability, in the convex hull of the surviving
subset of $S$, and thus of the full set $S$. We refer to this application of
    the exponential mechanism in the appropriate affine subspace as the
    \emph{base case} of the recursion.

    The points of $S\cap f$ are not on a standard grid within $f$. (They lie of course in the
    standard uniform grid $G$ of side length $1/X$ within the full-dimensional cube, but $G\cap f$ is not a
    standard grid and in general has a different, coarser resolution.) 
    We overcome this issue by noting that
    there always exist $j$ coordinates,
which, without loss of generality, we assume to be $x_1,\ldots,x_j$, such that $f$ can be expressed
in parametric form by these coordinates. We then
    project $f$ (and $S\cap f$) onto the $x_1x_2\cdots x_j$-coordinate subspace $f'$.
    We recurse within $f'$, where the projected points of $S\cap f$ do lie in a standard grid
(a cross-section of $G$), and then lift the output point $x_0'$,
    which lies, with high probability, in $\conv(S_0')$, back to a point $x'\in f$.
    It is straightforward to verify that if $x_0'$ is in
    the convex hull of the projected points then $x'$ is in the convex hull of $S\cap f$.

    \subsection{Finding an affine subspace with many points privately}
    \label{sec:details}

    For every affine subspace $f$, of dimension $0\le j\le d-1$, spanned by some subset
    of (at least) $j+1$ points of $G$, we denote by $c(f)$ the number of points
    of $S$ that $f$ contains, and refer to it as the \emph{size} of $f$.

    We start by computing $c(f)$ for every subspace $f$ spanned by points of $S$,
    as follows. We construct the (potentially degenerate) arrangement $\A(S^*)$ of the set $S^*$ of the
hyperplanes dual to the points of $S$. During this construction,
    we also compute the multiplicity of each flat in this arrangement, namely,
    the number of dual hyperplanes that contain it. Each intersection flat of
    the hyperplanes is dual to an affine subspace $f$ defined by the corresponding
    subset of the primal points of $S$ (that it contains), and $c(f)$ is the number of dual
    hyperplanes containing the flat. In other words, as a standard byproduct of
the construction of $\A(S^*)$, we can compute the sizes of all the
    affine subspaces that are spanned by points of $S$, in $O(n^d)$ overall time.

    We define
    \[
    M_i = M_i(S) = \max \{c(f) \mid  \text{ $f$ is spanned by a subset of $S$ and is of dimension } \allowbreak \text{{\em at most} $i$}\},
    \] and compute $M'_i = M_i + Y_i$, where $Y_i$ is a random variable drawn (independently
    for each $i$) from a Laplace distribution with parameter $b := \frac{1}{\eps}$
    centered at the origin. (That is, the probability density function of $Y_i$ is
    $\frac{1}{2b}e^{-|x|/b} = \frac{\eps}{2}e^{-\eps|x|}$.)

    Our algorithm now uses a given confidence parameter $\beta\in (0,1)$ and proceeds as follows.
    If for every $j=0,\ldots,d-1$
\begin{equation}
    M'_j \le  n - (d-j+1)k - \frac{1}{\eps}\log\frac{2}{\beta} , \label{eq:rul1}
\end{equation}
    we apply the base case.
    Otherwise, set $j$ to be the smallest index such that
\begin{equation}
    M'_j >  n - (d-j+1)k - \frac{1}{\eps}\log\frac{2}{\beta} .
\label{eq:rul2}
\end{equation}

    Having fixed $j$, we find (privately) a subspace $f$ of dimension $j$
that contains a large number of points of $S$. To do so, let $Z_j$ be the collection
of all $j$-dimensional subspaces that are spanned by $j+1$ affinely independent
points of the grid $G$ (not necessarily of $S$). We apply the exponential mechanism
to pick an element of $Z_j$, by
    assigning a {\em score} $s(f)$ to each subspace of $Z_j$, which we set to be
    \[
    s(f) = \max \left\{0, c(f) - M_{j-1}  \right\} \ ,
    \]
    if $j \ge 1$, and $s(f)=c(f)$ if $j=0$.
    Note that by its definition, $s(f)$ is zero if $f$ is not spanned by points of $S$.
Indeed, in this case the $c(f)$ points contained in $f$ span some subspace of dimension
$\ell \le j-1$ and therefore $M_{j-1}$ must be at least as large as $c(f)$. We will shortly argue that
$s(f)$ has sensitivity at most $2$ (Lemma \ref{lem:sen-s}), and thus conclude that this step preserves the differential
privacy of the procedure.

    We would like to apply the exponential mechanism as stated above in time proportional
to the number of subspaces of non-zero score, because this number depends only on $n$
(albeit being exponential in $d$) and not on (the much larger) $X$. However, to normalize
the scores to probabilities, we need to know the number of elements of $Z_j$ with zero score,
or alternatively to obtain
the total number of subspaces spanned by $j+1$ points of $G$ (that is, the size of $Z_j$).

    We do not have a simple expression for $|Z_j|$
(although this is a quantity that can be computed, for each $j$, independently of $S$, once and for all),
    but clearly $|Z_j|\le X^{d{(j+1)}}$.
    We thus augment $Z_j$ (only for the purpose of analysis)
with $X^{d{(j+1)}}-|Z_j|$ ``dummy'' subspaces, and denote the augmented set by $Z_j'$, whose
cardinality is now exactly $X^{d{(j+1)}}$.
    We draw a subspace $f$ from $Z_j'$ using the exponential mechanism.
To do so we need to compute, for each score $s\ge 0$, the number $N_s$ of elements of
$Z'_j$ that have score $s$, give each such element weight $e^{\eps s/4}$, choose the index
$s$ with probability proportional to $N_s e^{\eps s/4}$, and then choose  uniformly a subspace from
 those with score $s$. It is easy to check that this is indeed an implementation of the exponential mechanism as described in Section \ref{sec:exp}.
 
    If the drawing procedure decides to pick a subspace that is not spanned by points of $S$,
or more precisely decides to pick a subspace of score $0$, we stop the whole procedure, with a failure.
If the selected score is positive, the subspace to be picked is spanned by $j+1$ points of $S$,
and those subspaces are available (from the dual arrangement construction outlined above). We thus obtain
a selected subspace $f$ (by randomly choosing an element from the available list of these subspaces), and apply the algorithm  recursively within $f$, restricting the input to $S\cap f$. (Strictly speaking, as noted above, we apply the algorithm to a projection of $f$ onto a suitable
$j$-dimensional coordinate subspace.)

    It follows that we can implement the exponential mechanism on all subspaces in $Z_j'$
in time proportional to the number of subspaces spanned by points of $S$, which is $O(n^d)$, and therefore
the running time of this subspace selection procedure (in each recursive call) is $O(n^d)$.

\subsubsection{Privacy analysis}

    \begin{lemma} \label{lem:senM}
        Let $S_1= S_0 \cup \{x\}$ and $S_2=S_0 \cup \{y\}$ be two neighboring data sets.
        Then, for every $i = 0,\ldots,d-1$, we have $|M_i(S_1)-M_i(S_2)|\le 1$.
    \end{lemma}
    \begin{proof}
    Let $f$ be a subspace of dimension at most $i$ that is spanned by points of $S_1$
    and contains $M_i(S_1)$ such points.
    If $f$ does not contain $x$ then $f$ is
    also  a candidate for $M_i(S_2)$, so
    in this case $M_i(S_2) \ge M_i(S_1)$. 
    If $f$ does contain $x$ then $S_1\cap f \setminus \{x\} \subseteq S_0$ spans a subspace
    $f'$ of $f$ which is  of dimension at most $i$, so it
    is a candidate for $M_i(S_2)$. Since it does not contain $x$ (and may contain $y$) we have in this case that 
  $M_i(S_2) \ge M_i(S_1)-1$. Therefore we can conclude that in anycase $M_i(S_2) \ge M_i(S_1)-1$. 
   A symmetric argument shows that 
    $M_i(S_1) \ge M_i(S_2) - 1$. Combining these two inequalities the lemma follows.
    \end{proof}

    \begin{lemma} \label{lem:private}
        The value of each $M'_i$, for $i=0,\ldots,d-1$, is $\eps$-differentially private.
    \end{lemma}
    \begin{proof}
    This follows from  standard arguments in differential privacy (e.g., see
    \cite{DR14,S17}), since, by Lemma~\ref{lem:senM}, $M_i$ is of \emph{sensitivity} $1$
    (in the sense shown in the lemma).
    \end{proof}

Since we choose $j$ by comparing each of the $M'_j$'s to a value which is the same for neighboring data sets $S_1$ and $S_2$ (which have the same cardinality $n$),
Lemma~\ref{lem:private} implies that the choice of the dimension $j$ is differentially private. 

The next step is to choose the actual subspace $f$ in which to recurse. The following lemma implies
that this step too is differentially private. 
    \begin{lemma} \label{lem:sen-s}
        Let $S_1$ and $S_2$ be as in Lemma~\ref{lem:senM}. Then, for each $j=0,\ldots,d-1$ and for every subspace
$f \in Z'_j$, we have $|s_{S_1}(f)-s_{S_2}(f)| \le 2$.
    \end{lemma}
    \begin{proof}
Fix $j$ and $f\in Z'_j$. Clearly, $|c_{S_1}(f)-c_{S_2}(f)| \le 1$, and, by
    Lemma~\ref{lem:senM}, $M_{j-1}$ is also of \emph{sensitivity} $1$, and the claim follows.
    \end{proof}

    \begin{lemma} \label{privacy}
        The subspace-selection procedure described in this section (with all its recursive calls) is
        $2d^2\eps$-differentially private.
    \end{lemma}
    \begin{proof}
    By Lemma \ref{lem:private} the computation of
    each $M'_i$ is $\eps$-differentially private, and by Lemma~\ref{lem:sen-s} the exponential mechanism
on the scores $s(f)$ is also $\eps$-differentially private. Since we compute at most $d$ values $M'_i$ at each step, and
we recurse at most $d$ times,
the claim follows by composition~\cite{DR14,S17}.
    \end{proof}

\begin{remark}
We can save a factor of $d$
in Lemma \ref{privacy} by using a framework called {\em the sparse vector technique},
see e.g., \cite{DR14}.
\end{remark}

    \subsubsection{Utility analysis}
    The following lemma is the key for our utility analysis.

    \begin{lemma} \label{lem:count}
    Let $\beta\in (0,1)$ be a given parameter. For  $k = \frac{n}{4d}$ and
    for every  $j=0,\ldots,d-1$ the following properties hold. 
    
\smallskip
\noindent    
    (i) If $M_j \ge n - (d-j+1)k$ then, with probability
    at least $1-\beta$,
    \[
    M'_j \ge n - (d-j+1)k - \frac{1}{\eps}\log\frac{2}{\beta} .
    \]
    (ii) On the other hand, if
    $
    M_j \le n - (d-j+1)k - \frac{2}{\eps}\log\frac{2}{\beta} ,
    $
    then, with probability at least $1-\beta$,
    \[
    M'_j \le n - (d-j+1)k - \frac{1}{\eps}\log\frac{2}{\beta} .
    \]
    \end{lemma}
    \begin{proof}
        (i) follows since the probability of the Laplace noise $Y_j$ to be
        smaller than $-\frac{1}{\eps}\log\frac{2}{\beta}$ is at most $\beta$,
 and (ii) follows since the probability of $Y_j$ to be
        larger than $\frac{1}{\eps}\log\frac{2}{\beta}$ is
        also at most $\beta$.
    \end{proof}

    We say that (the present recursive step of) our algorithm {\em fails} if one of the
    complements of the events specified in Lemma \ref{lem:count}
    happens, that is, the step fails if for some $j$, either (i)
    $M_j \ge n - (d-j+1)k$ and $M'_j < n - (d-j+1)k - \frac{1}{\eps}\log\frac{2}{\beta}$,
    or (ii) $M_j \le n - (d-j+1)k - \frac{2}{\eps}\log\frac{2}{\beta}$ and
    $M'_j > n - (d-j+1)k - \frac{1}{\eps}\log\frac{2}{\beta}$.
Otherwise we say that (this step of) our algorithm {\em succeeds}.\footnote{Note that there is a `grey zone' of values of $M_j$ between these two bounds, in which the step always succeeds.}

    It follows from Lemma \ref{flat} that if the algorithm succeeds and applies the base case then $D_{\ge k}$
    is full dimensional. Furthermore, if the algorithm succeeds and decides to recurse on a
subspace of dimension $j$ (according to the rule in~(\ref{eq:rul1}) and~(\ref{eq:rul2})) then, for every $\ell < j$,
    $M_\ell \le n - (d-\ell+1)k$ and $M_j \ge n - (d-j+1)k - \frac{2}{\eps}\log\frac{2}{\beta}$.
    The following lemma is an easy consequence of this reasoning. 
    
    \begin{lemma}
        If the algorithm succeeds, with dimension $j$, and applies the exponential mechanism to pick
a $j$-dimensional subspace, then there exists a $j$-dimensional subspace $f$ with score
        $s(f) = M_j - M_{j-1} \ge k-\frac{2}{\eps}\log\frac{2}{\beta}$. Furthermore, if
        $k\ge \frac{8d^2}{\eps}\log X + \frac{8}{\eps}\log \frac{1}{\beta}$ then,
 with probability at least $1-\beta$, the exponential mechanism picks a subspace $f$ with
$s(f) \ge M_j - M_{j-1} - \frac{k}{2} \ge \frac{k}{2}-\frac{2}{\eps}\log\frac{2}{\beta}$.
    \end{lemma}
    \begin{proof}
    The first part of the lemma follows from the definition of success, as just argued.
    For the second part notice that, since $|Z_j'|\le X^{d^2}$, the probability of drawing
a subspace $f'\in Z_j'$ of score smaller than $M_j - M_{j-1} - \frac{k}{2}$ is at most
\[
X^{d^2} \cdot 
 \frac{ 
	e^{\eps(M_j-M_{j-1}- k/2)/{4}}}{ 
	e^{\eps(M_j-M_{j-1})/{4}}} = X^{d^2} \cdot e^{-\eps k/8} \ .
\]
    This expression is at most $\beta$ if
            $k\ge \frac{8d^2}{\eps}\log X + \frac{8}{\eps}\log \frac{1}{\beta}$.
    \end{proof}

    If follows that if our algorithm succeeds and recurses in
    a subspace $f$ of dimension $j$ then, with probability at least $1-\beta$,
    \begin{align*}
    c(f) & \ge M_{j-1} + s(f) \ge M_j - \frac{k}{2} \\
    &\ge n - (d-j+1)k  - \frac{2}{\eps}\log\frac{1}{\beta} - \frac{k}{2} \ge
    n - \left(d-j+\frac{3}{2}\right)k - \frac{2}{\eps}\log\frac{1}{\beta} .
    \end{align*}

That is, when we recurse in $f$ of dimension $j$ we lose at most
    $\left(d-j+\frac{3}{2}\right)k + \frac{2}{\eps}\log\frac{1}{\beta}$ points.
Denote by $d_0=d,d_1,\ldots,d_t$ the sequence of dimensions into which the procedure recurses
(reaching the base case at dimension $d_t\ge 0$). Hence, keeping $k$ fixed throughout the recursion,
at the $r$-th recursive step we lose at most
    $\left(d_r-d_{r+1}+\frac{3}{2}\right)k + \frac{2}{\eps}\log\frac{1}{\beta}$ points.
    Summing up these losses over $r=0,\ldots,t-1$, the total loss is at most
    \[
(d_0-d_t)k + \frac{3}{2}kt
    + \frac{2t}{\eps}\log\frac{1}{\beta}
    \le \frac{5d}{2} \cdot k
    + \frac{2d}{\eps}\log\frac{1}{\beta} .
    \]

    Substituting $k=\frac{n}{4d}$,
     we get that the total number of points that we loose is
   at most $\frac{2n}{3}$ if
    $n = \Omega\left( \frac{d}{\eps}\log \frac{1}{\beta} \right)$, with a sufficiently
large constant of proportionality.

Notice that  we keep $k$ fixed throughout the recursion
and $n$ may decrease. Consequently,
 if $n'$ is the number of points in some recursive call
in some dimension $\ell < d$,
then $n' \ge \frac{n}{3}$ and therefore $k = \frac{n}{4d} \le \frac{3n'}{4d}$ which is still smaller than the centerpoint guarantee of $\frac{n'}{\ell+1}$.

 As described, our subspace-selection procedure (with all its recursive calls) is $2 d^2 \eps$-differentially private. Dividing $\eps$ by $2 d^2$ we get that our subspace-selection procedure is $\eps$-differentially private, and that the total number of points we lose is
    much smaller than $n$ if
    $n = \Omega\left( \frac{d^3}{\eps}\log \frac{1}{\beta} \right)$. 

    Recall Section \ref{sec:base}, where we showed that we need
    $n = \Omega\left(\frac{d^4\log dX}{\eps} \right)$ for the ($\eps$-differentially private) base case to work correctly.
    (Recall also that the base case is actually applied in a suitable projection of the
    terminal subspace onto some coordinate-frame subspace of the same dimension, and that
    the above lower bound on $n$ suffices for any such lower-dimensional instance too.)

    The following theorem summarizes the result of this section.

    \begin{theorem} \label{thm:pure}
     Given $n = \Omega\left(\frac{d^4\log dX}{\eps} + \frac{d^3 \log \frac{1}{\beta}}{\eps} \right)$ points,
        our algorithm (including all recursive calls and the base case) is
        $\eps$-differentially private,
        runs in $O\left( n^{1+(d-1) \lfloor d/2 \rfloor} \right)$ time,
        and finds a point of depth at least
        $k=\frac{n}{4d}$ with probability at least $1-2d^2\beta$.
    \end{theorem}
    \begin{proof}
        The privacy statement follows by composition, using Lemma~\ref{privacy}
        and the privacy properties of the exponential mechanism.
The confidence bound follows since the probability of failure
        of the recursive call in a subspace of dimension $j$
        is at most $(j+1) \beta$.
The running time of the algorithm is dominated by the running time of the
exponential mechanism that we perform at the bottom (base case) of the recursion.
    \end{proof}

    \section{An $O(n^d)$ algorithm via volume estimation}\label{sec:nd_time}

    The upper bound on the running time in Theorem \ref{thm:pure} is dominated by the running time of
    the base case, in which we compute all the regions $D_{\ge \ell}$ explicitly, which takes $n^{O(d^2)}$ time.
    To reduce this cost, we use instead a mechanism that (a) estimates the volume of $D_{k}$
    to a sufficiently small relative error, and (b) samples a random point ``almost'' uniformly from $D_k$.
    We show how to accomplish (a) and (b) using the volume estimation mechanisms of Lov\'asz and Vempala~\cite{LV:vol} and later of
Cousins and Vempala \cite{CV}. We also show how to use these procedures to implement
    approximately the exponential mechanism described in Section \ref{sec:base}.
    These algorithms are Monte Carlo, so they fail with some probability, and when they fail we may
lose our $\eps$-differential privacy guarantee. As a result, the modified algorithm will not be purely
differentially private, as the one in Section \ref{uni-gen}, and we will only be able to guarantee that
it is $(\eps,\delta)$-differentially private, for any prescribed $\delta>0$. 
We obtain the following theorem which we prove in the rest of this section.
    \begin{theorem} \label{thm:approx}
	Given $n = \Omega\left(\frac{d^4\log \frac{dX}{\delta}}{\eps}\right)$
	points,
    our algorithm (including all recursive calls and the base case) is
    $(\eps,\delta)$-differentially private,
    runs in $O\left(  n^{d} \right)$ time,
    and  finds a point of depth at least
    $k=\frac{n}{4d}$ with probability at least $1-\delta$.
\end{theorem}


\subsection{Volume estimation}

We use the following result of Cousins and Vempala for estimating the volumes of the convex polytopes $D_{\ge k}$.

\begin{theorem}[Cousins and Vempala \protect{\cite[Theorem 1.1]{CV}}] \label{thm:cv}
	Let $K$ be a convex body in $\reals^d$ that
	contains the unit ball around the origin and satisfies $E_K(||X||^2)=O(d)$.\footnote{Here $E_K$ is expectation under the uniform measure within $K$.}
	Then, for any $\alpha , \beta > 0$, there is an algorithm that can access $K$ only via membership queries
	(each determining whether a query point $q$ lies in $K$), and, with probability at least $1-\beta$,
	approximates the volume of $K$ with a relative error $\alpha$ (i.e., in the range $1\pm \alpha$ times the true volume),
	and performs ${\displaystyle O\left(\frac{d^3}{\alpha^2}\log^2 d \log^2 \frac{1}{\alpha} \log^2 \frac{d}{\alpha} \log \frac{1}{\beta}\right)}$
	membership queries in $K$.
\end{theorem}

We do not know when the algorithm fails; failure simply means that the value that it returns is not within
the asserted range. 

\paragraph*{Membership tests.}
Given a query point $q$ and a parameter $\ell$, to test whether $q\in D_{\ge \ell}$
amounts to verifying that $q$ lies in each of the halfspaces whose intersection
forms $D_{\ge \ell}$. Lemma \ref{lem:Dd-1} shows that there are at most $O(n^{d-1})$
such halfspaces, and we can prepare in advance a list of these halfspaces, as a by-product of the construction of
the arrangement $\A(S^*)$ of the hyperplanes dual to the points in $S$. 
Thus we can implement a membership query in $O(n^{d-1})$ time.
(Recall that we hide in the $O(\cdot)$-notation factors that are polynomial in $d$.)

In order to bring $D_{\ge \ell}$ into the `rounded' form assumed in 
Theorem \ref{thm:cv}, we use an algorithm of Lov\'asz and Vempala~\cite{LV:vol}
to bring a convex body $K$ into an \emph{approximate isotropic position}.

\begin{definition}
	A convex body $K$ in $\reals^d$ is in isotropic position if for every unit vector $v\in \reals^d$,
	$\int_K (v\cdot x) dx =  1$. $K$ is in $t$-nearly isotropic position if for
	every unit vector $v\in \reals^d$, $\frac{1}{t}\le \int_K (v\cdot x) dx \le t$.
\end{definition}

It is easy to verify that if $K$ is in $t$-nearly isotropic position then $E_K(||X||^2)\le td$.
\begin{theorem}[Lov\'asz and Vempala~\protect{\cite[Theorem 1.1]{LV:vol}}] \label{thm:lv}
	There is an algorithm that, given a convex body $K$ that contains the unit ball around the origin
	and is contained in a ball of radius $R$, 
	and a confidence parameter $\beta > 0$,
	brings $K$ into a 2-nearly isotropic position.
	The algorithm succeeds with probability at least $1-\beta$, and performs
	$O(d^4 \log^7 d \log \frac{d}{\beta} \log R)$ membership queries in $K$.
\end{theorem}

We can use some of our observations in Section \ref{sec:base} to show:
\begin{lemma}
	Each of the polytopes $D_{\ge \ell}$ of nonzero volume
	contains a ball of radius  at least
$\frac{1}{(d+1)\left( d^{(d+1)/2}X^d \right)^{d(d+1)+1}}$.
\end{lemma}
\begin{proof}
	As we showed in Section \ref{sec:base}, each vertex of $D_{\ge \ell}$ has rational coordinates of
	common integral denominator of value at most $d^{d(d+1)/2}X^{d^2}$.
	We take
	$d+1$ linearly independent vertices of $D_{\ge \ell}$, and denote by $z$ their center of mass, which is guaranteed to be inside $D_{\ge \ell}$. 
	This center of mass $z$
 has rational coordinates with a common integer denominator of value at most
	$(d+1)(d^{d(d+1)/2}X^{d^2})^{(d+1)}$.
	
	We also showed in Section \ref{sec:base} that each hyperplane $h$ defining $D_{\ge \ell}$
	has integer coefficients of absolute value at most $d^{d/2}X^{d}$.

	The largest radius of a ball enclosed in $D_{\ge \ell}$ and centered at $z$ is the minimum distance of $z$ from the defining (supporting) hyperplanes of $D_{\ge\ell}$. The distance from any such hyperplane $h$ is positive (since $z$ lies in the interior of $D_{\ge\ell}$), and is a rational whose denominator is the common denominator of the coordinates of $z$ times
 the Euclidean norm of the vector of coefficients of $h$, excluding the free term, which
	is at most $\sqrt{d} d^{d/2}X^{d}$. That is, this distance is 
 at most
	$$
 \frac{1}{(d+1)d^{d(d+1)^2/2}X^{d^2(d+1)}\cdot d^{(d+1)/2}X^{d}}=\frac{1}{(d+1)\left( d^{(d+1)/2}X^d \right)^{d(d+1)+1}}\ .
	$$
	 This implies the lemma.
\end{proof}

$D_{\ge \ell}$ is contained in the unit cube, and therefore it is contained in a ball of radius $\sqrt{d}$.
It follows that we can scale $D_{\ge \ell}$ so that it satisfies the conditions of Theorem \ref{thm:lv},
with $R = {(d+1)^{3/2}\left( d^{(d+1)/2}X^d \right)^{d(d+1)+1}}$. Hence we have
$\log R = O(d^3(\log X + \log d))$.

Our entire volume estimation algorithm applies Theorem \ref{thm:lv} to $D_{\ge \ell}$ and
then Theorem \ref{thm:cv} to the `rounded' polytope. By the union bound it fails with probability at most $2\beta$.
The running time is larger by a factor of  $O(d^2)$ than the number of membership queries.

\subsection{Implementing the exponential mechanism}

We now show how to use the estimates of the volumes of the polytopes $D_{\ge \ell}$ to sample
``almost'' uniformly from a region $D_{\ell}$, which is picked with probability $\mu'_\ell$, where
\[
\frac{1-\alpha}{1+\alpha}\mu_\ell \le \mu'_\ell \le \frac{1+\alpha}{1-\alpha}\mu_\ell ,\qquad\text{and}\qquad
\mu_\ell = \frac{e^{\eps \ell/2}{\rm Vol}(D_\ell)}
{\sum_{j\ge 0} e^{\eps j/2} {\rm Vol}(D_j)} ,
\]
as defined in Section \ref{sec:base}.

We carry out this modified sampling by defining new probabilities $\lambda'_j$,
for $j=0,\ldots,\td_\mx(S)$ (so that $\sum_j \lambda'_j = 1$),
choosing the region $D_{\ge \ell}$ (rather than $D_\ell$) with probability $\lambda'_\ell$
and then sampling almost uniformly a point in the chosen region $D_{\ge \ell}$, using the
sampling algorithm in \cite{CV}, whose properties are  stated in the following theorem.
In the theorem, the \emph{variation distance} between two measures $\mu$, $\nu$
is $\max_B |\mu(B) - \nu(B)|$, over all possible events $B$.
\begin{theorem}
	[Cousins and Vempala \protect{\cite[Theorem 1.2]{CV}}]\label{thm:cv-samples}
	Let $K$ be a convex body in $\reals^d$ that
	contains the unit ball around the origin and satisfies $E_K(||X||^2)=O(d)$.
	Then, for any $\eta, \beta > 0$, there is an algorithm that can access $K$
	only via membership queries, and, with probability at least $1-\beta$, generates
	random points from a distribution $U'_K$ that is within total variation distance
	at most $\eta$ from the uniform distribution on $K$.
	The algorithm performs $O({d^3}\log d \log^2 \frac{d}{\eta} \log \frac{1}{\beta})$ membership queries
	for each random point that it generates. The running time is larger by a factor of  $O(d^2)$ than the number of membership queries.
\end{theorem}

\paragraph*{Exact sampling.}

Let us denote $V_\ell := {\rm Vol}(D_{\ge \ell})$ and $v_\ell := {\rm Vol}(D_{ \ell})$.
We first derive probabilities $\lambda_\ell$ such that if we pick
$D_{\ge \ell}$ with probability $\lambda_\ell$, and then pick a point uniformly from
$D_{\ge \ell}$, then we simulate the exponential mechanism (over the regions $D_\ell$) exactly.

The probability that $q\in D_m$, conditioned on having chosen $D_{\ge \ell}$
to sample from, is $v_m/V_\ell$ for $m\ge\ell$; it is $0$ for $m<\ell$.
Therefore, the unconditional probability that $q\in D_m$ is
${\displaystyle \left( \sum_{\ell\le m} \frac{\lambda_\ell}{V_\ell} \right) v_m}$,
and we want this probability to be proportional to $e^{\eps m/2}v_m$.
That is, we want all the equalities
${\displaystyle \sum_{\ell\le m} \frac{\lambda_\ell}{V_\ell} = C e^{\eps m/2}}$,
for $m\ge 0$, to hold, for a suitable normalizing parameter $C$.
Putting $x_\ell = \lambda_\ell/V_\ell$, we get the triangular system
${\displaystyle \sum_{\ell\le m} x_\ell = C e^{\eps m/2}}$,
for $m\ge 0$, which solves to
\[
x_0 = C, \qquad\text{and}\qquad
x_\ell = C\left( e^{\eps\ell/2} - e^{\eps(\ell-1)/2} \right) = C\left(1-e^{-\eps/2}\right) e^{\eps\ell/2} ,
\]
for $\ell \ge 1$. That is, to ensure that $\sum_{\ell \ge 0} \lambda_\ell = 1$, we then put
\[
C = \frac{1}{V_0 + \left(1-e^{-\eps/2}\right) \sum_{j \ge 1} e^{\eps j/2} V_j},\]
and \[ \lambda_0 = C V_0, \; \text{and }
\lambda_\ell = C\left(1-e^{-\eps/2}\right) e^{\eps\ell/2} V_\ell, \text{   for } \ell \ge 1 .
\]

\paragraph*{Approximate sampling.}
We approximate these probabilities, by replacing the exact volumes
$V_\ell$ by their approximate values $V'_\ell$ that we obtain from Theorems \ref{thm:lv} and \ref{thm:cv},
and get the corresponding approximations $C'$ to $C$ and $\lambda'_\ell$ to $\lambda_\ell$.
By construction, and by the probability union bound, we have, for the prescribed error parameter $\alpha< 1$,
with probability at least $1-2n\beta$,
${\displaystyle (1-\alpha)V_\ell \le V'_\ell \le (1+\alpha)V_\ell}$, for every $\ell=0,\ldots,\td_\mx(S)$,
which is easily seen to imply that
\begin{equation} \label{lambdap}
\frac{1-\alpha}{1+\alpha}\lambda_\ell \le \lambda'_\ell \le \frac{1+\alpha}{1-\alpha}\lambda_\ell ,\qquad\text{ for every $\ell$} .
\end{equation}

We sample $D_{\ge \ell}$ using the probabilities $\lambda'_\ell$, and then
sample from $D_{\ge \ell}$ using Theorem \ref{thm:cv-samples}.
Denote by $\Pr_{\rm exp}(B)$, for any measurable subset $B$ of $[0,1]^d$, the probability that the pure exponential mechanism,
that uses the (costly to compute) exact values $\lambda_\ell$, as just reviewed, returns a point in $B$.
It follows that the conditional probability $\Pr(B)$ of returning
a point in $B$, conditioned on all the parameters $\lambda'_\ell$ satisfying the inequalities (\ref{lambdap}),
and conditioned on successful sampling from $D_{\ge \ell}$ using Theorem \ref{thm:cv-samples},
satisfies
\begin{lemma} \label{lem:lambdap}
	\begin{equation} \label{eq:11}
	\frac{1-\alpha}{1+\alpha}\Pr_{\rm exp}(B) -\eta \le \Pr(B) \le \frac{1+\alpha}{1-\alpha}\Pr_{\rm exp}(B) +\eta \ .
	\end{equation}
\end{lemma}
\begin{proof}
	Indeed, by definition of
	our algorithm we have that
	\begin{equation} \label{eq-l0}
	\Pr(B) = \sum_{\ell\ge 0} \Pr\left(B | D_{\ge\ell}\right) \cdot \Pr(D_{\ge\ell}) ,
	\end{equation}
	where $D_{\ge\ell}$ is a mnemonic for the event that the mechanism decides to sample from the region $D_{\ge\ell}$.
	By Theorem \ref{thm:cv-samples} we have that
	\begin{equation} \label{eq-l2}
	\Pr_u\left(B | D_{\ge\ell}\right) -\eta \le
	\Pr\left(B | D_{\ge\ell}\right) \le \Pr_u\left(B | D_{\ge\ell}\right) +\eta \ ,
	\end{equation}
	where
	$\Pr_u\left(B | D_{\ge\ell}\right)$ is the probability
	to get a point in $B$ if we sample uniformly from $D_{\ge\ell}$. Combining Equations
	(\ref{eq-l0}) and (\ref{eq-l2}) we get that
	\begin{equation} \label{eq-l4}
	\sum_{\ell\ge 0}\Pr_u\left(B | D_{\ge\ell}\right) \cdot \Pr(D_{\ge\ell}) -\eta \le \Pr(B) \le \sum_{\ell\ge 0} \Pr_u\left(B | D_{\ge\ell}\right) \cdot \Pr(D_{\ge\ell}) +\eta \ .
	\end{equation}
	Now $\Pr(D_{\ge\ell})$  in our approximation scheme is $\lambda'_\ell$, so we get, using  (\ref{lambdap}), that
	\[
	\frac{1-\alpha}{1+\alpha} \sum_{\ell\ge 0} \lambda_\ell \Pr_u\left(B | D_{\ge\ell}\right) -\eta \le
	\Pr(B) \le \frac{1+\alpha}{1-\alpha} \sum_{\ell\ge 0} \lambda_\ell \Pr_u\left(B | D_{\ge\ell}\right) +\eta \ ,
	\]
	which implies the lemma by our argument  regarding the exact simulation of the exponential mechanism given above.
\end{proof}

If we set $\eta = \delta$, $\beta$ to be some constant multiple of $\frac{\delta}{n}$ (we divide by $n$ to account for the failure probablity over $D_{\ge \ell}$ for all $\ell$'s), and $\alpha$ to be a constant multiple of $\eps$ 
then it  follows from Equation (\ref{eq:11}) that our approximation of
the exponential mechanism is $(\eps, \delta)$-private.

It is also easy to verify that the utility of our approximate exponential mechanism is essentially the same as the utility analysis in Section \ref{sec:base}. This proves Theorem \ref{thm:approx}. \qed



\section{Conclusions}
We gave an $O(n^d)$-time algorithm for 
privately computing a point in the convex hull
of $\Omega(d^4  \log X)$ points with coordinates that are multiples
of $1/X$ in $[0,1]$.  
Even though this gives a huge improvement of what was  previously known and requires some nontrivial technical effort, and sophisticated sampling and volume estimation tools, this running time is still not satisfactory for large values of $d$. 
The main hurdle in improving it further is the 
nonexistence
of efficient algorithms for computing Tukey depths and Tukey levels. 

The main question that we leave open is whether there exists a differentially private algorithm for this task which is polynomial in $n$ and $d$ ? (when the input size, $n$, is still polynomial in $\log X $ and $d$).

\section*{Acknowledgments}
We thank Santosh Vempala for many helpful discussions.


\begin{thebibliography}{10}

\bibitem{BKN10}
Amos Beimel, Shiva~Prasad Kasiviswanathan, and Kobbi Nissim.
\newblock Bounds on the sample complexity for private learning and private data
  release.
\newblock In {\em TCC}, volume 5978 of {\em LNCS}, pages 437--454. Springer,
  2010.

\bibitem{BeimelMNS19}
Amos Beimel, Shay Moran, Kobbi Nissim, and Uri Stemmer.
\newblock Private center points and learning of halfspaces.
\newblock In {\em Conference on Learning Theory (COLT)}, pages 269--282, 2019.

\bibitem{BNS13b}
Amos Beimel, Kobbi Nissim, and Uri Stemmer.
\newblock Private learning and sanitization: Pure vs. approximate differential
  privacy.
\newblock In {\em APPROX-RANDOM}, volume 8096 of {\em LNCS}, pages 363--378.
  Springer, 2013.

\bibitem{BunDRS18}
Mark Bun, Cynthia Dwork, Guy~N. Rothblum, and Thomas Steinke.
\newblock Composable and versatile privacy via truncated cdp.
\newblock In {\em 50th Annual ACM SIGACT Symposium on Theory of Computing
  (STOC)}, pages 74--86, 2018.

\bibitem{BNSV15}
Mark Bun, Kobbi Nissim, Uri Stemmer, and Salil~P. Vadhan.
\newblock Differentially private release and learning of threshold functions.
\newblock In {\em IEEE 56th Annual Symposium on Foundations of Computer Science
  ({FOCS})}, pages 634--649, 2015.

\bibitem{Chaz}
Bernard Chazelle.
\newblock An optimal convex hull algorithm in any fixed dimension.
\newblock {\em Discrete Comput. Geom.}, 10:377--409, 1993.

\bibitem{CV}
Ben Cousins and Santosh~S. Vempala.
\newblock Gaussian cooling and ${O}^*(n^3)$ algorithms for volume and gaussian
  volume.
\newblock {\em {SIAM} J. Comput.}, 47(3):1237--1273, 2018.

\bibitem{DMNS06}
Cynthia Dwork, Frank McSherry, Kobbi Nissim, and Adam Smith.
\newblock Calibrating noise to sensitivity in private data analysis.
\newblock In {\em TCC}, volume 3876 of {\em LNCS}, pages 265--284. Springer,
  2006.

\bibitem{DR14}
Cynthia Dwork and Aaron Roth.
\newblock The algorithmic foundations of differential privacy.
\newblock {\em Found. Trends Theor. Comput. Sci.}, 9(3-4), 2014.

\bibitem{ESS}
Herbert Edelsbrunner, Raimund Seidel, and Micha Sharir.
\newblock On the zone theorem for hyperplane arrangements.
\newblock {\em {SIAM} J. Comput.}, 22(2):418--429, 1993.

\bibitem{KLMNS19}
Haim Kaplan, Katrina Ligett, Yishay Mansour, Moni Naor, and Uri Stemmer.
\newblock Privately learning thresholds: Closing the exponential gap.
\newblock {\em CoRR}, abs/1911.10137, 2019.
\newblock URL: \url{http://arxiv.org/abs/911.10137}.

\bibitem{Liu19}
Xiaohui Liu, Karl Mosler, and Pavlo Mozharovskyi.
\newblock Fast computation of Tukey trimmed regions and median in dimension {p
  $>$ 2}.
\newblock {\em J. of Comput. and Graph. Stat.}, 28(3):682--697, 2019.

\bibitem{LV:vol}
L{\'{a}}szl{\'{o}} Lov{\'{a}}sz and Santosh~S. Vempala.
\newblock Simulated annealing in convex bodies and an ${O}^*(n^4)$ volume
  algorithm.
\newblock {\em J. Comput. Syst. Sci.}, 72(2):392--417, 2006.

\bibitem{Matousek:2002}
Ji\v{r}\'i Matousek.
\newblock {\em Lectures on Discrete Geometry}.
\newblock Springer-Verlag, Berlin, Heidelberg, 2002.

\bibitem{MT07}
Frank McSherry and Kunal Talwar.
\newblock Mechanism design via differential privacy.
\newblock In {\em 48th Annual IEEE Symposium on Foundations of Computer Science
  ({FOCS})}, pages 94--103, 2007.

\bibitem{RR98}
Peter~J. Rousseeuw and Ida Ruts.
\newblock Constructing the bivariate {T}ukey median.
\newblock {\em Statistica Sinica}, 8(3):827--839, 1998.

\bibitem{Tukey}
John~W. Tukey.
\newblock Mathematics and the picturing of data.
\newblock In {\em Proc. of the International Congress of Mathematicians},
  volume~2, page 523–531, 1975.

\bibitem{S17}
Salil Vadhan.
\newblock The complexity of differential privacy.
\newblock In Yehuda Lindell, editor, {\em Tutorials on the Foundations of
  Cryptography: Dedicated to Oded Goldreich}, pages 347--450. Springer, 2017.

\end{thebibliography}
\end{document}